\pdfoutput=1 

\RequirePackage{snapshot}

\RequirePackage{sty/bootstrap}
\AppendInputPath{{sty/}{sty/img/}{fig/}{build/}}

\RequirePackage[twocolumn]{IEEEtranSetup}
\documentclass[conference,final]{IEEEtran}

\ProvideIf[true]{draft}
\IfClassLoadedWithOptionsTF{IEEEtran}{final}{\draftfalse}{}

\ProvideIf[false]{ieee}  
\ProvideIf[false]{ieeepreview}
\ProvideIf[false]{blind} 
\ProvideIf[true]{arxiv} 

\PassOptionsToPackage{compress}{boilerplate}
\ifieee 
  \usepackage[hyperref=false,bibtex]{boilerplate}
\else
  \ifarxiv
    \usepackage[biblatex]{boilerplate}
    
  \else
    \usepackage[bibtex]{boilerplate}
  \fi
\fi

\usepackage[capitalize]{cleveref}
\crefname{equation}{}{}
\Crefname{equation}{Equation}{Equations}

\NewDocumentCommand \arxivOrcid {sm} {%
  \IfBooleanTF{#1}{,\hspace{-3pt}}{}%
  \ifarxiv\orcidlink{#2}\fi%
}

\title{%
  Deterministic Identification \\ Over Multiple-Access Channels
}
\date{\today}
\author{
\ifblind\else
  \IEEEauthorblockN{%
    Johannes~Rosenberger\arxivOrcid*{0000-0003-2267-3794}
    Abdalla~Ibrahim\arxivOrcid*{0009-0008-7646-6276}
    Christian~Deppe\arxivOrcid*{0000-0002-2265-4887}
    Roberto~Ferrara\arxivOrcid{0000-0002-1991-3286}
  }%
  \IEEEauthorblockA{
    TUM School of Computation, Information and Technology,
    Technical University of Munich,
    \\
    Email: \{\tumail{johannes.rosenberger}%
      , \tumail{abdalla.m.ibrahim}%
      , \tumail{christian.deppe}%
      , \tumail{roberto.ferrara}%
    \}@tum.de
  }
\fi
}

\makeatletter
\hypersetup{
  pdftitle = \string\@title,
  pdfauthor = {},
  pdfkeywords = {},
  pdflang = {English},
}
\makeatother

\newcommand{\code}{\cC}
\NewDocumentCommand{\defname}{om}{\emph{#2}\IfValueT{#1}{~(#1)}}

\newcommand{\capacity}{\mathsf{C}}
\newcommand{\capT}{\capacity_\textsf{T}}
\newcommand{\capDI}{\capacity_\textsf{DI}}

\DeclareMathOperator{\Bern}{Bern}
\newcommand{\costIA}{\frac{1}{n} \sum_{i=1}^n \bm\phi}

\usepackage{tikz}
\newif\ifqtikz\qtikzfalse
\usetikzlibrary{math}
\usetikzlibrary{fpu}

\pgfkeys{/pgf/fpu}
\pgfkeys{/pgf/fpu/output format=fixed}
\tikzmath{
  function fold(\p, \q) { return \p*(1-\q) + \q*(1-\p); };
  function Hbin(\p) { return - \p * log2(\p) - (1 - \p)*log2(1-\p); };
  function Ibsc(\p, \q) { return Hbin(fold(\p, \q)) - Hbin(\q)); };
  function Ibec(\p, \q) {
    return - \p*(1-\q) * log2(\p*(1-\q)) - (1-\p)*(1-\q) * log2((1-\p)*(1-\q)) - \p*\q * log2(\p*\q);
  };
  function IgaussC(\s) { return log2(1 + \s); };
  function IgaussR(\s) { return IgaussC(\s)/2; };
  function IzChannel(\p, \d) { return Hbin(\p * \d) - \p * Hbin(\d); };
  function zChannelPopt(\d) {
    real \g;
    let \g = (1-\d)^((1-\d)/\d);
    return \g/(1 + \d*\g);
  };
}
\pgfkeys{/pgf/fpu=false}

\providecommand \ifqtikz {\iftrue}

\usetikzlibrary{fpu} 
\usetikzlibrary{math} 

\ifqtikz

\usetikzlibrary{math}
\usetikzlibrary{fpu}

\pgfkeys{/pgf/fpu}
\pgfkeys{/pgf/fpu/output format=fixed}
\tikzmath{
  function fold(\p, \q) { return \p*(1-\q) + \q*(1-\p); };
  function Hbin(\p) { return - \p * log2(\p) - (1 - \p)*log2(1-\p); };
  function Ibsc(\p, \q) { return Hbin(fold(\p, \q)) - Hbin(\q)); };
  function Ibec(\p, \q) {
    return - \p*(1-\q) * log2(\p*(1-\q)) - (1-\p)*(1-\q) * log2((1-\p)*(1-\q)) - \p*\q * log2(\p*\q);
  };
  function IgaussC(\s) { return log2(1 + \s); };
  function IgaussR(\s) { return IgaussC(\s)/2; };
  function IzChannel(\p, \d) { return Hbin(\p * \d) - \p * Hbin(\d); };
  function zChannelPopt(\d) {
    real \g;
    let \g = (1-\d)^((1-\d)/\d);
    return \g/(1 + \d*\g);
  };
}
\pgfkeys{/pgf/fpu=false}

\tikzset{
  font=\sffamily,
  rateBound/identification/.style = {color=identification},
  rateBound/transmission/.style = {dashed,color=transmission}
}

\fi

\newcommand \rateRegionBinAdderModTwo [2][scale=5]{

\providecolor{identification}{named}{blue}
\providecolor{transmission}{named}{red}

\pgfkeys{/pgf/fpu}
\pgfkeys{/pgf/fpu/output format=fixed}
\tikzmath{
 real \ticklen, \pcrossover, \capBSC;
 \ticklen = 0.03; 
 \pcrossover = #2;
 \capBSC = Ibsc(0.5, \pcrossover);
}
\pgfkeys{/pgf/fpu=false}

\begin{tikzpicture}[auto,thick,#1]
  \draw[->] (0,0) -- (0, {1+2*\ticklen}) node[above] {$R_2$};
  \draw[->] (0,0) -- ({1+2*\ticklen}, 0) node[right] {$R_1$};

  \draw (0, 1) -- +(-\ticklen, 0) node[left] {$1$};
  \draw (1, 0) -- +(0, -\ticklen) node[below] {$1$};
  \draw (0, \capBSC) -- (-\ticklen, \capBSC) node[above left,yshift=1.5em,rotate=90] {$1-H_2(\pcrossover)$};
  \draw (\capBSC, 0) -- (\capBSC, -\ticklen) node[below left,xshift=1.2em] {$1-H_2(\pcrossover)$};

  \draw[rateBound/identification] (0, 1) -- (1,1) node[below left] {DI} -- (1, 0);
  \draw[rateBound/transmission] (0, \capBSC) -- (\capBSC, 0) node[midway,sloped,below] {Transmission};
\end{tikzpicture}
}

\ifqtikz
\rateRegionBinAdderModTwo{0.1}
\fi

\providecommand \ifqtikz {\iftrue}

\usetikzlibrary{fpu} 
\usetikzlibrary{math} 

\ifqtikz

\usetikzlibrary{math}
\usetikzlibrary{fpu}

\pgfkeys{/pgf/fpu}
\pgfkeys{/pgf/fpu/output format=fixed}
\tikzmath{
  function fold(\p, \q) { return \p*(1-\q) + \q*(1-\p); };
  function Hbin(\p) { return - \p * log2(\p) - (1 - \p)*log2(1-\p); };
  function Ibsc(\p, \q) { return Hbin(fold(\p, \q)) - Hbin(\q)); };
  function Ibec(\p, \q) {
    return - \p*(1-\q) * log2(\p*(1-\q)) - (1-\p)*(1-\q) * log2((1-\p)*(1-\q)) - \p*\q * log2(\p*\q);
  };
  function IgaussC(\s) { return log2(1 + \s); };
  function IgaussR(\s) { return IgaussC(\s)/2; };
  function IzChannel(\p, \d) { return Hbin(\p * \d) - \p * Hbin(\d); };
  function zChannelPopt(\d) {
    real \g;
    let \g = (1-\d)^((1-\d)/\d);
    return \g/(1 + \d*\g);
  };
}
\pgfkeys{/pgf/fpu=false}

\tikzset{
  font=\sffamily,
  rateBound/identification/.style = {color=identification},
  rateBound/transmission/.style = {dashed,color=transmission}
}

\fi

\newcommand \rateRegionBinAdderModThreeNoiseless [1][scale=5]{

\providecolor{identification}{named}{blue}
\providecolor{transmission}{named}{red}

\pgfkeys{/pgf/fpu}
\pgfkeys{/pgf/fpu/output format=fixed}
\tikzmath{
 real \ticklen, \pcrossover, \capBSC;
 \ticklen = 0.03; 
}
\pgfkeys{/pgf/fpu=false}

\begin{tikzpicture}[auto,thick,#1]
  \draw[->] (0,0) -- (0, {1+2*\ticklen}) node[above] {$R_2$};
  \draw[->] (0,0) -- ({1+2*\ticklen}, 0) node[right] {$R_1$};

  \draw (0, 1) -- +(-\ticklen, 0) node[left] {$1$};
  \draw (1, 0) -- +(0, -\ticklen) node[below] {$1$};
  \draw (0, 0.5) -- (-\ticklen, 0.5) node[left] {$0.5$};
  \draw (0.5, 0) -- (0.5, -\ticklen) node[below] {$0.5$};

  \draw[dotted, thin] (0, 0.5) -- (1, 0.5) (0.5, 0) -- (0.5, 1);

  \draw[rateBound/identification] (0, 1) -- (1,1) node[below left] {DI} -- (1, 0);
  \draw[rateBound/transmission] (0.5, 1) -- (1, 0.5) node[midway,below left] {Transmission};
\end{tikzpicture}
}

\ifqtikz
\rateRegionBinAdderModThreeNoiseless
\fi

\providecommand \ifqtikz {\iftrue}

\usetikzlibrary{fpu} 
\usetikzlibrary{math} 

\ifqtikz

\usetikzlibrary{math}
\usetikzlibrary{fpu}

\pgfkeys{/pgf/fpu}
\pgfkeys{/pgf/fpu/output format=fixed}
\tikzmath{
  function fold(\p, \q) { return \p*(1-\q) + \q*(1-\p); };
  function Hbin(\p) { return - \p * log2(\p) - (1 - \p)*log2(1-\p); };
  function Ibsc(\p, \q) { return Hbin(fold(\p, \q)) - Hbin(\q)); };
  function Ibec(\p, \q) {
    return - \p*(1-\q) * log2(\p*(1-\q)) - (1-\p)*(1-\q) * log2((1-\p)*(1-\q)) - \p*\q * log2(\p*\q);
  };
  function IgaussC(\s) { return log2(1 + \s); };
  function IgaussR(\s) { return IgaussC(\s)/2; };
  function IzChannel(\p, \d) { return Hbin(\p * \d) - \p * Hbin(\d); };
  function zChannelPopt(\d) {
    real \g;
    let \g = (1-\d)^((1-\d)/\d);
    return \g/(1 + \d*\g);
  };
}
\pgfkeys{/pgf/fpu=false}

\tikzset{
  font=\sffamily,
  rateBound/identification/.style = {color=identification},
  rateBound/transmission/.style = {dashed,color=transmission}
}

\fi

\newcommand \rateRegionBinMul [2][scale=5]{

\providecolor{identification}{named}{blue}
\providecolor{transmission}{named}{red}

\pgfkeys{/pgf/fpu}
\pgfkeys{/pgf/fpu/output format=fixed}
\tikzmath{
 real \ticklen, \pcrossover, \capBSC;
 \ticklen = 0.03; 
 \pcrossover = #2;
 \capBSC = Ibsc(0.5, \pcrossover);
}
\pgfkeys{/pgf/fpu=false}

\begin{tikzpicture}[auto,thick,#1]
  \draw[->] (0,0) -- (0, {1+2*\ticklen}) node[above] {$R_2$};
  \draw[->] (0,0) -- ({1+2*\ticklen}, 0) node[right] {$R_1$};

  \draw (0, 1) -- +(-\ticklen, 0) node[left] {$1$};
  \draw (1, 0) -- +(0, -\ticklen) node[below] {$1$};

  \draw (0, \capBSC) -- (-\ticklen, \capBSC) node[above left,yshift=1.5em,rotate=90] {$1-H_2(\pcrossover)$};
  \draw (\capBSC, 0) -- (\capBSC, -\ticklen) node[below left,xshift=1.2em] {$1-H_2(\pcrossover)$};

  \draw[rateBound/identification] (0, 1) -- (1,0) node[midway,sloped,below] {DI};
  \draw[rateBound/transmission] (0, \capBSC) -- (\capBSC, 0) node[midway,sloped,below] {Transmission};
\end{tikzpicture}
}

\ifqtikz
\rateRegionBinMul{0.1}
\fi

\colorlet{identification}{blue}
\colorlet{transmission}{red}
\tikzset{
  font = {\sffamily\footnotesize}, 
  rateBound/identification/.style = {color=identification},
  rateBound/transmission/.style = {dashed,color=transmission}
}

\begin{document}

\ifonecol
\begingroup
\def\cmd#1{\texttt{\textbackslash{}#1}}
\fxwarning[author=JR]{
  The \textsf{fixme} package is loaded, so you can make editorial
  notes with \texttt{\cmd{fxnote}[author=<author>]\\\{<margin comment>\}},
  \texttt{\cmd{fxnote*}[author=<author>]\\\{<comment>\}\{<body text to highlight>\}},
  or, for notes that have to be resolved before submission,
  \cmd{fxfatal} and \cmd{fxfatal*} are used like \cmd{fxnote}.
  In the end of the document, there is also a list of corrections.
  The final submission should get the class option \texttt{final}. This
  hides edit notes and the list of corrections,
  and it produces a compile error for each remaining \cmd{fxfatal} calls.
}
\fxnote[author=JR]{
  Please add shorthands to \texttt{local-shorthands.tex}.
}
\endgroup
\fi

\begin{anfxnote}[author=JR]{Macro conventions}
  \section*{README}
  \label{README}
  \raggedright
  Please do not use \verb|\mathbf|, which does not work for all (math) symbols,
  but instead use the command \verb|\bm|, which is more flexible.
  Mixing both leads to inconsistent output,
  e.g. \verb|\mathbf{X, \alpha}| $\mapsto \mathbf{X, \alpha}$,
  while \verb|\bm{X, \alpha}| $\mapsto \bm{X, \alpha}$.
  The files \verb|sty/shorthands.sty| and \verb|local-shorthands.tex|
  define many shorthands, e.g.
  
  \begin{tabular}{ll}
    \verb|\bX| & $\bX$ \\
    \verb|\bbx| & $\bbx$ \\
    \verb|\cX| & $\cX$ \\
    \verb|\bcX| & $\bcX$ \\
    \verb|\bbN| & $\bbN$ \\
    \verb|\bbR| & $\bbR$ \\
    \verb|\expect| & $\expect$ \\
    \verb|\One| & $\One$ \\
    \verb|\ind{a}| & $\ind{a}$ \\
    \verb|\set{a}| & $\set{a}$ \\
    \verb|\tup{a}| & $\tup{a}$ \\
    \verb|\paren{a}| & $\paren{a}$ \\
    \verb|\intv{a}| & $\intv{a}$ \\
    \verb|\brack{a}| & $\brack{a}$ \\
    \verb|\abs{a}| & $\abs{a}$ \\
    \verb|\card{\cX_1}| & $\card{\cX_1}$ \\
    \verb|\code| & $\code$ \\
    \verb|\capacity| & $\capacity$ \\
    \verb|\capDI| & $\capDI$.
  \end{tabular}
  
  \noindent
  Please define your own shorthands in \verb|local-shorthands.tex|.
\end{anfxnote}

\maketitle


\begin{abstract}
  Deterministic identification over
  $K$-input multiple-access channels with average input cost constraints
  is considered.
  The capacity region for deterministic identification
  is determined for an average-error criterion,
  where arbitrarily large codes are achievable.
  For a maximal-error criterion,
  upper and lower bounds on the capacity region are derived.
  The bounds coincide if all average partial point-to-point channels are
  injective under the input
  constraint, i.e. all inputs at one terminal are mapped to distinct
  output distributions, if averaged over the inputs at all other terminals.
  The achievability is proved by treating the MAC as an arbitrarily varying
  channel with average state constraints.
  For injective average channels,
  the capacity region is a hyperrectangle.
  The modulo-2 and modulo-3 binary adder MAC are presented as examples of
  channels which are injective under suitable input constraints.
  The binary multiplier MAC is presented as an example of a non-injective
  channel, where the achievable identification rate region
  still includes the Shannon capacity region.
\end{abstract}

\section{Introduction}
One of the most persistent aspects in modern communication systems is the ever-increasing need for higher rates for a variety of complex and interrelated tasks. In some tasks and applications,
\defname[ID]{identification} of messages~\cite{jaja1985identification,ahlswedeDueck1989id1,ahlswede2021identification_probabilistic_models} is a potential alternative communication paradigm to the usual data transmission
as studied by Shannon~\cite{shannon1948it0}.
Shannon's model requires each receiver in a multi-user communication system to decode their own
received signals and to estimate which
message was encoded by the sender.
In ID, the receiver is only interested in answering the
“Yes or No”-question
whether a specific message was encoded or not,
performing only a binary hypothesis test.
One surprising consequence is that due to the simplification
to a hypothesis testing problem, ID code sizes
can be doubly exponential in the block length.
However, this scaling is only achieved by ambiguous encoding using randomness.
Using deterministic encoders, the number of messages is upper-bounded
by the number of codewords
\cite{salariseddighPeregBocheDeppe2022det_ID_powerConstraints_tit,ahlswedeCai1999detID},
which scales exponentially for discrecte channels, as in Shannon's transmission paradigm.
Nevertheless, the achievable rates strictly exceed the Shannon capacity for
noisy memoryless channels.
This was first observed in~\cite{jaja1985identification}
for the binary symmetric channel.
Over continuous channels, \defname[DI]{deterministic identification} code sizes scale in the
order of $n^n$~\cite{SJPBDS23, SJ21, SP20}.
Super-exponential codes can also be achieved for DI if
local randomness
can be generated using resources such as feedback
\cite{feedback} or sensing \cite{sensing},
making it effectively a randomized ID.
Applications are expected mainly in
authentication tasks~\cite{moulin2001watermarking_information_theory,steinbergMerhav2001id_watermarking,steinberg2002id_watermarking,ahlswedeCai2006watermarking} and event-triggered systems.
Possible applications include also molecular
communication~\cite{SJPBDS23}
and further 6G applications~\cite{CB21}.

We consider the task of
deterministic (non-randomized) ID with $K$ senders and one receiver via a
\defname[$K$-MAC]{multiple-access channel},
a channel with $K$ inputs and one output.
A MAC can be used to model cooperative senders but also to model jamming attacks by a malevolent sender,
in which case the MAC is usually called an arbitrarily-varying channel~\cite{lapidothNarayan1998uncertainty}.
The MAC was first analyzed for transmission in~\cite{A71, L72},
and surveys are given in~\cite{M88} and~\cite[Chapter 4]{elgamalKim2011network_it}.
The randomized ID capacity of the $K$-MAC was determined in~\cite{ahlswede2008gtid_updated}.

Here, we determine the DI capacity region of the $K$-MAC
under an average-error criterion,
where arbitrarily large codes are possible,
since codewords can be multiply assigned.
Under a maximal-error criterion,
we derive upper and lower bounds on the capacity region
for DI over the $K$-MAC,
with average input cost constraints.
The bounds coincide if all partial point-to-point
channels are injective, if averaged over the input types
for the other channel inputs. Then, the capacity region
is a hyperrectangle in $\bbR^K$.
To prove the main result, we derive a lower bound on the
achievable rates over arbitrarily-varying channels with
state constraints.
We compare the DI and transmission capacity regions
of binary adder and multiplier channels.
In general, for the MAC, the very constrained task of DI
can be solved more efficiently than message transmission.


\section{Setup and preliminaries}

We consider DI over a
\defname[$K$-MAC]{discrete memoryless $K$-input multiple-access channel}.
There are $K$ senders, where each sender
$k \in [K] \coloneqq \set{1,\dots,K}$
emits a word $x^n_k \in \cX_k^n$ of letters from an input alphabet $\cX_k$.
At the receiver, the channel outputs a word $y^n \in \cY^n$ of letters from
the output alphabet $\cY$.
Bold font denotes sequences, where each element corresponds to one sender.
The tuple of all inputs is denoted by
$\bbx^n = (x_1^n,\dots,x_K^n) \in \bcX^n \coloneqq \cX_1^n \times \dots \times \cX_K^n$,
where $\bcX = \cX_1 \times \dots \times \cX_K$,
and for all $i \in [n]$,
we denote by $\bbx_i \coloneqq (x_{1,i},\dots, x_{K,i})\in \bcX$
the tuple sent with the $i$-th use of the $K$-MAC.
The $K$-MAC is described by
a conditional probability distribution $W_{Y|\bX} : \bcX \to \cP(\cY)$,
where $\cP(\cY)$ is the set of all finite probability distributions over $\cY$.
Given an input $\bbx^n \in \bcX^n$,
the discrete memoryless (DM) $K$-MAC produces an output
$y^n \in \cY^n$
according to the product distribution
\(
  W_{Y|\bX}^n(y^n|\bbx^n) = \prod_{i=1}^n W_{Y|\bX}(y_i|\bbx_i).
\)
A channel $W : \cX \to \cP(\cY)$ is called \defname{injective} if
it maps different inputs to different output distributions, i.e.\ if $W(\cdot|x) \ne W(\cdot|x')$ for all distinct $x, x' \in \cX$.

\subsection{Deterministic identification codes}
\label{sec:preliminiaries.DIcode}

Let $\bM \in \bbN^K$, where $M_k$ is the
number messages of the $k$-th sender, be a code size tuple.
An \defname{$(\bM, n)$ deterministic-identification (DI) code} for a $K$-MAC
$W_{Y|\mathbf{X}} : \bcX \to \cP(\cY)$
is a pair $\code = (\bbf, g)$, where 
$\bbf$ is a tuple of encoding functions
$f_k : [M_k] \to \cX_k^n$ for each $k \in [K]$,
and $g : [M_1] \times \dots \times [M_K] \times \cY^n \to \set{0,1}^K$
is an identification function.
Each sender $k \in [K]$ selects one message $m_k \in [M_k]$,
and for the message tuple $\bbm = (m_1,\dots, m_K)$,
the codeword tuple $\bbx^n(\bbm) = \bbf(\bbm) = \tup{f_1(m_1),\dots,f_K(m_K)}$
is sent over the channel.
At the channel output, the receiver observes the noisy sequence
$Y^n \sim W_{Y|\bX}^n\tup{ \cdot \,\middle|\, f(\bbm) }$.
Suppose that the receiver is interested in identifying the tuple
$\bbm' = (m'_1,\dots ,m'_K)$,
for each $k \in [K]$ he performs a binary hypothesis test: He declares that $m'_k$
was sent if $g_k(\bbm', Y^N) = 1$,
where $g_k(\bbm', Y^N)$ denotes the $k$-th component of $g(\bbm', Y^N)$;
otherwise, he declares that $m'_k$ was not sent.
Thus, the receiver decides correctly regarding sender $k$ if
$
  g_k(\bbm', Y^n) = \ind{ m_k' = m_k }.
$
Conversely, the error probability is given by
%
\begin{gather}
  e_k(\bbm'|\bbm)
    = 
    W^n_{Y|\bX} \tup[\big]{ g_k(\bbm', Y^n) = \ind{ m'_k \ne m_k } \,\big|\, \bbf(\bbm) }.
\end{gather}
Sending often incurs a cost, e.g. for transmission power.
This can be modeled by a non-negative cost function $\phi_k : \cX_k \to \bbR_+$ and a non-negative cost constraint $\Phi_k \in \bbR_+$
for each sender $k \in [K]$, where $\bbR_+$ are the non-negative real numbers.
An \defname{$(\bM, n, \bm\phi, \bm\Phi, \lambda)$ DI code} for a $K$-MAC $W_{Y|\bX}$
is defined as an $(\bM, n)$ DI code with cost functions $\bm\phi = (\phi_1,\dots,\phi_K)$ and cost constraints
$\bm\Phi = (\Phi_1,\dots,\Phi_K) \in \bbR_+^K$,
such that 
\begin{itemize}
    \item
    the tuple $\bbx^n = \bbf(\bbm)$ satisfies the average input constraint
      $
      \costIA(\bbx_i)
        \preceq \Phi_k,
      $
    where 
    $\bm\phi(\bbx_i) \coloneqq \tup{ \phi_1(x_{1,i}),\dots,\phi_K(x_{K,i}) }$,
    and “$\preceq$” holds iff “$\le$” holds for all
    components;
    \item
    the total error probability
    $e(\bbm'|\bbm) = W^n_{Y\bX}(\forall\,k: g_k(\bbm', Y^n) = \ind{m_k' \ne m_k} | \bbf(\bbm))$
    is bounded by 
    $
      e(\bbm'|\bbm)
      \le \lambda,
      $
    for all $\bbm,\bbm' \in [M_1] \times \dots \times [M_K]$.
\end{itemize}
The DI rate of user $k$ is defined as $R_k = \frac{1}{n} \log M_k$.
A rate tuple $\bR \in \bbR_+^K$ is 
\defname{achievable} over $W_{Y|\bX}$
under an average input constraint
if, for every $\lambda > 0$ and all sufficiently large $n$,
there exists an $(\bM, n, \bm\phi, \bm\Phi, \lambda)$ DI code for $W_{Y|\bX}$.
The \defname{DI capacity region} $\capDI(W_{Y|\bX}, \bm\phi, \bm\Phi)$
is the closure of the set of all achievable rate pairs.

\subsection{Laws of large Numbers}

We use basic concepts of typicality and the method of types,
as defined by Kramer~\cite[Chapter 1]{kramer2008multi_user_book}
and Csiszár--Körner~\cite[Chapter 2]{csiszarKoerner2011IT}.
%
The $n$-type $\hat{P}_{x^n}$ of a sequence \(x^n \in \cX^n\)
is defined by $\hat{P}_{x^n}(a) = \frac{1}{n} \sum_{i=1}^n \ind{x_i=a}$.
The set of all $n$-types over a set $\cX$ is denoted by
  {$\cP_n(\cX) \subset \cP(\cX)$}.
Furthermore, given a PMF \(p_X \in \cP(\cX)\), the \(\epsilon\)-typical set
is defined as
\begin{gather}
\label{eq:def.typSet}
  \cT_\epsilon^n(p_X) = \set[\big]{x^n :
    \abs[\big]{\hat{P}_{x^n}(a) - p_X(a)} \leq \epsilon \, p_X(a), \,\forall\,a }
    \,.
\end{gather}
The jointly typical set
$\cT_\epsilon^n(p_{XY})$ is defined similarly, where $\cX$ is replaced by
$\cX \times \cY$ and $p_X$ by $p_{XY}$.
Given a sequence \(x^n \in \cX^n\), the \emph{conditionally} \(\epsilon\)-typical set
is defined as
\(\cT_\epsilon^n(p_{XY} | x^n) = \set{y^n \in \cY^n : (x^n, y^n) \in \cT_\epsilon^n(p_{XY})}\).

\subsection{Previous results}


DI over point-to-point channels has been studied
for binary alphabets~\cite{jaja1985identification}
and for general, finite alphabets~\cite{ahlswedeCai1999detID,salariseddighPeregBocheDeppe2022det_ID_powerConstraints_tit}.
Already in~\cite{ahlswedeDueck1989id1}, it was observed that the DI capacity
of a discrete memoryless channel (a $K$-MAC with $K=1$)
equals the logarithm of the number of distinct rows
of the channel transition matrix.
In~\cite{ahlswedeCai1999detID}, an achievable bound for arbitrarily-varying
channels was derived, i.e. for 2-MACs where only one partial channel is used
for identification, while the other channel input is controlled by a
jammer.
%
%
In~\cite{salariseddighPeregBocheDeppe2022det_ID_powerConstraints_tit}, the following capacity theorem for discrete memoryless channels was derived.

\begin{theorem}
    [{see~\cite[Theorem~4]{salariseddighPeregBocheDeppe2022det_ID_powerConstraints_tit}}]
  The DI capacity of a DMC $W$ under the input constraint $(\phi, \Phi)$
  is given by
  \begin{gather}
    \capDI(W, \phi, \Phi) = \capDI(W^\star, \phi, \Phi)
    \,,
  \end{gather}
  where $W^\star$ is the maximizing injective channel that
  can be obtained from $W$ by restricting
  the channel input alphabet, i.e. removing $x'$ from $\cX$,
  if $W(\cdot|x) = W(\cdot|x')$, such that for all
  $x$ in the reduced alphabet $\cX^\star$, $W^\star(\cdot|x) = W(\cdot|x)$.
\end{theorem}

\section{Results}

Our main results are for a maximal error criterion,
as used in the definition of a DI code
in the previous section,
whereas \cref{sec:results.avg} modifies the problem
to an average error criterion.

\subsection{Maximal error criterion}

An achievable bound for the $K$-MAC can be obtained by
interpreting every partial point-to-point channel
as an \defname[AVC]{arbitrarily varying channel}~\cite{lapidothNarayan1998uncertainty},
where the channel varies for every input symbol
depending on a channel state.
In our case, the state comprises all channel inputs
by the other users.
An AVC is usually a worst case model,
where a malicious jammer can select the state
to inhibit communication.
It can be applied to the MAC, where every user tries
\emph{not} to interfere with the others,
by constraining the channel states to the codebook.
The following proposition provides an achievable
DI rate for an AVCs under input and state type constraints.

\begin{proposition}
\label{prop:achiev}
  Let $W : \cX \times \cS \to \cP(\cY)$ bet a $2$-MAC, $P \in \cP_n(\cX)$ and $\cP_S \subseteq \cP_n(\cS)$, and define the set
  \begin{gather}
    \cQ(P, \cP_S , W) = \set[\big]{
      p_{XX'SY} \in \cP(\cX \times \cX \times \cS \times \cY) : \nonumber
      \\p_S, p_{S'} \in \cP_S,
      \ p_{Y|X} = p_S W,
      \ p_{Y|X'} = p_{S'} W, \nonumber
      \\p_X = p_{X'} = P,
      \ \text{$X - X' - Y$ is a Markov chain}
    }\,.
  \end{gather}
  Then, for every rate
  \begin{gather}
    R \ge \min_{p_{XX'SY} \in \cQ(P, \cP_S, W)} I(X'; XSY)
  \end{gather}
  there exists a pair $(f, g)$ and a constant $\tau > 0$ such that 
  \begin{subthms}
  \subthm
  $f(m) \in \cT_0^n(P)$, for every message $m \in \intv{2^{nR}}$;
  \subthm 
  for every $p_S \in \cP_S$ and every $s^n \in \cT_0^n(p_S)$,
  $(f, g)$ is a $(2^{nR}, n, e^{-n\tau})$ DI-code for the channel $W^n(\cdot | \cdot, s^n)$.
  \end{subthms}
\end{proposition}
The proof follows in Section~\ref{proof:prop:achiev}.

\begin{remark}
  \label{remark:avc}
For injective DMCs, noise noes not reduce the
DI capacity
(see~\cite{ahlswedeDueck1989id1,salariseddighPeregBocheDeppe2022det_ID_powerConstraints_tit}), and
by \cref{prop:achiev} the same is true for
a 2-MAC, where only one message is identified, if
$p_S W(\cdot|x) \ne p_{S'} W(\cdot|x')$ for all $p_S,p_{S'} \in \cP_S$ and all distinct $x, x' \in \cX$,
since in this case, $H(X'|X) = 0$, and $I(X'; XSY) = H(X)$.
Hence, increasing the noise does not degrade performance.
By the following capacity theorem, the rate given by
\cref{prop:achiev} gives indeed the optimal rates
for the $K$-MAC.
\end{remark}

Define $S_k = (X_1,\dots,X_{k-1},X_{k+1},\dots,X_K)$ and the sets
\begin{align}
  \cR^l(\bX, Y)
  &=
  \set{
    \begin{array}{l}
      \bR \in \bbR^K :
      \forall\,k \in [K],\,
      \\
      0 \le R_k \le
      \hspace{-1em}
      \min\limits_{\substack{p_{X_k X'_k S_k Y} \\ \in \cQ(p_{X_k}, \set{ p_{S_k}} , W)}}
      \hspace{-1em}
      I(X'_k; \bX Y)
    \end{array}
  }
  ,
  \\
  \cR^u(\bX)
  &=
  \set{ \bR \in \bbR^K : 0 \le R_k \le H(X_k),\, \forall\,k \in [K] }
  \,.
\end{align}

\begin{theorem}
\label{thm:capacityDM_MAC}
  The DI-capacity region of a $K$-MAC
  $W_{Y|\bX} : \bcX \to \cP(\cY)$ under an input cost constraint
  $\costIA(\bbx_i) \preceq \bm\Phi$
  satisfies
  \begin{gather}
    \bigcup_{\substack{
        p_{\bX} \in \cP(\bcX) : \\
        \expect[\bm\phi(\bX)] \preceq \bm\Phi
    }}
    \hspace{-1.1em}
    \cR^l(\bX, Y)
    \subseteq
    \capDI(W_{Y|\bX}, \bm\phi, \bm\Phi)
    \subseteq
    \hspace{-1.1em}
    \bigcup_{\substack{
        p_{\bX} \in \cP(\bcX) : \\
        \expect[\bm\phi(\bX)] \preceq \bm\Phi
    }}
    \hspace{-1.1em}
    \cR^u(\bX)
    \,.
  \end{gather}
  The bounds coincide if
  every average partial channel $W_{Y|X_k} = p_{S_k}W_{Y|\bX}$
  is injective, under the input constraint.
\end{theorem}
\begin{remark}
For continuous $\bm\phi$, the bounds coincide even if all but
finitely many average partial channels are injective,
since $\cR^l$ and $\cR^u$ are continuous and
the capacity region is the closure of all achievable rate regions.
\end{remark}
\begin{proof}
Note that for codewords
$x^n \in \cT_0^n(p_{X_k})$ and $X_k \sim p_{X_k}$,
\begin{gather}
  \expect[\phi_k(X_k)] = \frac{1}{n} \sum_{n=1}^N \phi_k(x_i).
  \label{eq:typicalAverage}
\end{gather}
Thus, achievability of the lower bound follows by applying
\cref{prop:achiev} individually to every partial channel,
by assigning $X = X_k$ such that $\expect[\phi_k(X_k)] \le \Phi_k$,
and $\cP_S = \set{ p_S }$, where
$S = S_k = (X_1,\dots,X_{k-1},X_{k+1},\dots,X_K)$,
for every $k \in [K]$.
Thus, for some $\tau > 0$,
every $k \in [K]$,
and arbitrary $s_k^n \in p_{S_k}$,
we get a $(2^{n R_k}, n, \phi_k, \Phi_k, e^{-n\tau})$
DI code $(f_k, g_k)$
for the channel $W_{Y|X_k S_k}^n(\cdot|\cdot, s_k^n)$,
where $R_k \ge \min_{p_{X_k X'_k S_k Y} \in \cQ(p_{X_k}, \set{ p_{S_k} }, W)}
I(X'_k; X_k S_k Y)$.
By combining the point-to-point codes, we obtain the
$(\bM, n)$ DI code $(\bbf, g)$, where
$\bbf = [f_1,\dots,f_K]$, $g(\bbm, Y^n) = [g_1(m_1),\dots,g_K(m_K)]$,
$\bM = [2^{n R_1}, \dots, 2^{n R_K}]$,
and the error probabilities
are bounded by $e_k(\bbm'|\bbm) \le 2^{-n\tau}$.
The types satisfy the input constraint, and
by the union bound, the total error probability
satisfies
$
  e(\bbm'|\bbm) \le \sum_{k=1}^K 2^{-n\tau} = K 2^{-n\tau},
$
which tends to zero as $n \to \infty$.
Therefore, for sufficiently large $n$, $(\bbf, g)$
is an $(\bM, n, \bm\phi, \bm\Phi, e^{-n \tau})$ DI code for the
$K$-MAC $W_{Y|\bX}$, and in the limit $n \to \infty$,
it achieves the maximal rates in $\cR(p_{\bX})$.

The proof of the upper bound follows the same line of argument
as the point-to-point converse
in~\cite[Section III.D]{salariseddighPeregBocheDeppe2022det_ID_powerConstraints_tit}:
Without loss of generality we consider Sender 1.
Suppose the same codeword is assigned to
two messages $m_1, m'_1 \in \cM_1$, $m_1 \neq m_1'$,
i.e. $f_1(m_1) = f_1(m_1')$,
and let $(m'_2,\dots,m'_K) = (m_2,\dots,m_K)$.
Then, given that $e_k(\bbm'|\bbm') \le \frac{1}{2}$,
\begin{align}
  e_k(\bbm'|\bbm)
    & = W_{Y|\bX}^n \tup{ g_k(\bbm', Y^n) = 1 | \bbf(\bbm) }  \\& = W_{Y|\bX}^n \tup{ g_k(\bbm', Y^n) = 1 | \bbf(\bbm') }
  \\& = 1 - W_{Y|\bX}^n \tup{ g_k(\bbm', Y^n) = 0 | \bbf(\bbm') }
  \\& = 1 - e_k(\bbm'|\bbm')
  \\& \ge \frac{1}{2}
  \,,
\end{align}
and hence, no rate is achievable.
Consequently, the number of messages is at most the number of codewords,
for each sender $k \in [K]$.
By~\eqref{eq:typicalAverage},
the type $p_{X_k}$ of every codeword satisfying
the input constraint also
satisfies $\expect[\phi_k(X_k)] \leq \Phi_k$.
Hence,
\begin{align}
  R_k
    & \le \max_{p_{X_k} \::\: \expect[\phi_k(X_k)] \le \Phi_k}
    \frac{1}{n} \log_2\tup{ \card{\cP_n(\cX_n)} \cdot \abs{\cT_0^n(p_{X_k})} }
  \\& \le \max_{p_{X_k} \::\: \expect[\phi_k(X_k)] \le \Phi_k} H(X_k) + \frac{1}{n} \log_2 \card{\cP_n(\cX_k)}
  \\& \le \max_{p_{X_k} \::\: \expect[\phi_k(X_k)] \le \Phi_k} H(X_k) + \frac{1}{n} \card{\cX_k} \log_2 (n+1)
  \,,
\end{align}
which converges to $H(X_k)$ as $n \to \infty$.
\cref{thm:capacityDM_MAC} follows,
becuase for injective average partial channels,
$H(X'_k|X_k) = 0$ and thus $I(X'_k; X_k S_k  Y) = H(X_k)$ (see \cref{remark:avc}).
\end{proof}

\subsection{Average error criterion}
\label{sec:results.avg} 

In \cref{sec:preliminiaries.DIcode} we defined a maximum-error criterion.
If all messages are equiprobable, we might consider
the average error instead. A simple
expurgation (see also~\cite[Proposition~2]{hanVerdu1992idNewResults})
shows that given $\bbm' = \bbm$, this
does not change the capacity.
But given $m'_k \ne m_k$,
the average error probability is given by
\fxwarning*[author=JR]{Can be simpler?}{
\begin{gather}
  \bar e_k(\bbm')
  =
  \frac{1}{M_k - 1} \prod_{\substack{k' = 1 \\ k' \ne k}}^K
  \frac{1}{M_{k'}}
  \sum_{k'' = 1}^K
  \sum_{\substack{m_{k''} = 1 \\ m_{k''} \ne m_k}}^{M_{k''}}
  e_k(\bbm'|\bbm).
  \label{eq:err.af.dmc}
\end{gather}
}
We define an $(\bM, n, \bm\phi, \bm\Phi, \lambda)$ DIA code
as an $(\bM, n)$ DI-code 
satisfying the input cost constraint $\costIA(\bbx_i) \preceq \bm\Phi$
and the error bounds
$e(\bbm'|\bbm) \le \lambda$,
for $\bbm' = \bbm$,
and
$\bar e(\bbm')
= \sum_{k : m'_k \ne m_k} e_k(\bbm') \le \lambda$,
for $\bbm' \ne \bbm$.
for the false-DI error.
Since code sizes can scale super-exponentially for
\defname[DIA]{average-error DI},
we call a sequence $(\bM_n)_{n\in \bbN}$
of code size tuples
\defname{achievable} if, for every $\lambda > 0$
and sufficiently large $n$, there exists
an $(\bM_n, n, \bm\phi, \bm\Phi, \lambda)$ DIA code
(see~\cite{salariseddighPeregBocheDeppe2022det_ID_powerConstraints_tit}
for a discussion on scalings).

\begin{theorem}
  \label{thm:err.af.mac.infinite}
  For every $K$-MAC $W_{Y|\bX}$,
  if $\capDI(W_{Y|\bX}, \bm\phi, \bm\Phi) > 0$,
  then arbitrary sequences $(\bM_n)_{n \in \bbN}$, are achievable under the constraint
  $\costIA(\bbx_i) \preceq \bm\Phi$.
\end{theorem}
\begin{proof}
  For the discrete memoryless channel (1-MAC),
  the theorem follows from the proof of
 ~\cite[Proposition~1]{hanVerdu1992idNewResults},
  without modification, by assuming
  deterministic encoding distributions and noting that
  the DIA code has the same codebook as the DI code,
  and hence the cost constraint is satisfied.
  The proof assigns arbitrarily large equivalence classes
  for the messages (e.g. by taking every message modulo $M$),
  and one codeword to each equivalence class.
  The error probabilities converge, since for every
  receiver's message $m'$, in is unlikely
  that $f(\tilde m) = f(m')$,
  on average over random $\tilde m$.

  The result for the $K$-MAC
  follows from time division:
  Consider an $(\bM, n, \bm\phi, \bm\Phi, \lambda)$
  DI code.
  For each $k$, there exists a letter $x_k \in \cX_k$
  that minimizes the cost $\phi_k(x_k)$.
  Assign a time slice of $n/K$ symbols to each sender.
  During the rest of the time, sender $k$ emits $x_k$, constantly.
  When it is $k^\star$'s turn to send,
  during his time slice, the channel is a DMC
  $W_{Y|X_{k^\star}} = W_{Y|\bX}(\cdot|x_1,\dots,x_{k^\star-1},\cdot,x_{k^\star+1},x_K)$.
  The maximal achievable rate is thus
  $R_k = \frac{1}{nK} \log\log M_k > 0$,
  and by the result for DMCs,
  every sequence of code sizes $M^A_{k,n}$ is achievable,
  i.e. for sufficiently large $n$,
  there exists an $(M^A_{k,n}, n, \phi_k, \Phi_k, \frac \lambda K)$
  DIA code $(f_k^A, g_k^A)$ for $W_{Y|X_{k^\star}}$.
  Let 
  $\bM^A_n = [M^A_{1,n},\dots,M^A_{K,n}]$,
  $\bbf^A = [f_1^A, \dots, f_K^A]$,
  and
  $g^A(\bbm', Y^n) = [g_1^A(m'_1, Y^n),\dots,g_K^A(m'_K, Y^n)]$.
  Then, $(\bbf^A, g^A)$ is an
  $(\bM^A_n, n, \bm\phi, \bm\Phi, \lambda)$ DIA code.
\end{proof}

\section{Examples}


\begin{figure}[t]
  \begin{minipage}{0.23\textwidth}
    \centering
    \rateRegionBinAdderModThreeNoiseless[scale=2.5]
    \caption{\label{fig1}
    DI capacity vs. Transmission capacity regions for the modulo-3 BA-MAC}
  \end{minipage}
  \hfill
  \begin{minipage}{0.23\textwidth}
    \centering
    \rateRegionBinAdderModTwo[scale=2.5]{0.05}
    \caption{\label{fig2}
    DI capacity vs. Transmission capacity regions for the modulo-2 BA-MAC and the noisy BM-MAC}
  \end{minipage}
\end{figure}

The following examples consider different types of binary symmetric 2-MAC
$W_{Y|X_1,X_2}$ with $\cX_1=\cX_2=\{0,1\}$.
The symmetry of the channel implies that $W_{Y|X_1,X_2=0}=W_{Y|X_1=0,X_2}$ and
$W_{Y|X_1,X_2=1}=W_{Y|X_1=1,X_2}$.

\subsection{Modulo-3 Binary Adder 2-MAC (BA-MAC)}

The modulo-3 binary adder channel is defined by
\begin{gather}
Y=(X_1+X_2) \mod 3 \,,
\end{gather}
where $\mathcal{Y}=\{0,1,2\}$.
The partial point-to-point channels are described by the following stochastic transition matrices:
\begin{align}
 W_{Y|X_1,X_2=0}
   &=
     \begin{bmatrix}
       1 & 0 & 0   \\
       0 & 1 & 0
     \end{bmatrix}
     , 
     &
 W_{Y|X_1,X_2=1}
   &= \begin{bmatrix}
        0 & 1 & 0 \\
        0 & 0 & 1
      \end{bmatrix}
     .
\end{align}
Clearly, every partial channel is injective on average,
as long as $p_{X_1}(1), p_{X_2}(2) \ne 0.5$.
Howevery, the input types may get arbitrarily close to the uniform
distribution.
Hence, by \cref{thm:capacityDM_MAC},
the DI capacity region is given by $R_1\leq 1$ and $R_2\leq 1$,
while the capacity region for transmission is defined by
$R_1\leq 1,\;
R_2\leq 1,\;
R_1+R_2\leq 1.5$.
The capacity regions are displayed in \cref{fig1}.
While both of DI and transmission here share the same corner points of the capacity region, the DI capacity region is markedly larger, taking a square shape, while transmission capacity region takes a triangular shape.
  
\subsection{Noisy Modulo-2 Binary Adder 2-MAC}

This channel is defined by
\begin{gather}
  Y = (X_1+X_2+Z) \mod 2 \,,
\end{gather}
where $Z \sim \Bern(q)$
and $\mathcal{Z}=\mathcal{Y}=\set{0,1}$.
The partial point-to-point channels are given by the following transition
matrices:
\begin{align}
  W_{Y|X_1,X_2=0}
  &=
    \begin{bmatrix}
      1-q & q  \\
      q & 1-q 
    \end{bmatrix}
    , \\
 W_{Y|X_1,X_2=1}
 &=
 \begin{bmatrix}
   q & 1-q  \\
   1-q & q 
 \end{bmatrix}
 .
\end{align}
By the symmetry of the channel, every partial channel is
injective on average, as for the Modulo-3 BA-MAC.
By \cref{thm:capacityDM_MAC}, the DI capacity region is
$\capDI = \set{ (R_1, R_2) : R_1 \le 1,\, R_2 \le 1 }$,
whereas the capacity region for transmission is
$\capT = \set{ (R_1, R_2) : R_1+R_2 \le 1-H_2(q) }$.
Clearly, $\capT \subseteq \capDI$,
as can be seen in \cref{fig2}, which visualizes the two capacity regions.
The corner points of $\capDI$ are outside of $\capT$, since the channel noise does not reduce the rate,
and the partial point-to-point channels are effectively orthogonal.
This is in contrast the optimal transmission scheme in which time-sharing is optimal.

\subsection{Noisy Binary Multiplier 2-MAC (BM-MAC)}

Here, we discuss a channel defined by
\begin{gather}
  Y=X_1 \cdot X_2+Z \mod 2 \,,
\end{gather}
where $Z \sim \Bern(q)$ and $\mathcal{Z}=\mathcal{Y}=\{0,1\}$.
The partial point-to-point channels are
given by the following transition
matrices:
\begin{align}
  W_{Y|X_1,X_2=0}
  &=
  \begin{bmatrix}
    1-q & q  \\
    1-q & q 
  \end{bmatrix}
  , \\
  W_{Y|X_1,X_2=1}
  &=
  \begin{bmatrix}
    1-q & q  \\
    q & 1-q 
  \end{bmatrix}
  .
\end{align}
$W_{Y|X_1,X_2=1}$ is injective, while $W_{Y|X_1,X_2=0}$ is not.
But if $q \ne 0.5$ and $p_{X_2}(0) < 1$, the average channel
$p_{Y|X_1} = p_{X_2} W_{Y|X_1X_2}$ is also injective.
Hence, $R_1 = \max_{p_{X_1}} H(X_1) = 1$ is achievable,
and by symmetry, the capacity region is the same as in \cref{fig2}.
Also, the transmission capacity region is
$\capT = \set{ (R_1, R_2) : R_1+R_2 \le 1-H_2(q) }$
(see \cref{fig2}), as for the Modulo-2 BA-MAC.

Note that, by adding a Hamming weight constraint
$
\frac{1}{n}\sum_{i=1}^{n}x_{2,i} \le \Phi_2
$
for the second channel input, the transmission
capacity of the first partial channel is reduced to
$\capT^{(1)} = \Phi_2 \cdot (1-H_2(q))$, while the DI
capacity is still $\capDI^{(1)} = 1$,
since the average channel is injective.

\section{Proof of \cref{prop:achiev}}
\label{proof:prop:achiev}

The proof is based on the proof of~\cite[Theorem~4]{ahlswedeCai1999detID}.
We add an average state constraint, similar to
\cite{csiszarNarayan1989avcClasses_capacity_decodingRules}.

Let $0 < R < \min_{p_{XX'SY} \in \cQ(p_X, \cP_S, W)} I(X'; XYS)$ and
\begin{align}
  \cQ_\delta(P, \cP_S, W) &= \set[\big]{
    p_{XX'SY} \in \cP(\cX \times \cX \times \cS \times \cY) :
    \nonumber\\&
    p_S, p_{S'} \in \cP_S,\,
    p_{Y|X} = p_S W,\, p_{Y|X'} = p_{S'} W,\,
    \nonumber\\&
    p_X = p_{X'} = P,\,
    I(X; Y | X' S) \le \delta
    }\,.
\end{align}
Then, $\cQ(p_X, \cP_S, W) = \bigcap_{\delta > 0} \cQ_\delta(p_X, \cP_S, W)$,
and by the continuity of the mutual information, there are
$\alpha, \delta > 0$ s.t.
\begin{gather} 
  R < \min_{p_{XX'SY} \in \cQ_\delta(p_X, \cP_S, W)} I(X'; XYS) - \alpha
  \,.
  \label{eq:rate}
\end{gather}

\subsection{Code construction}

By~\cite[Lemma~1]{csiszarKoerner1981avc_maxerr} or
\cite[Lemma~3]{csiszarNarayan1989avcClasses_capacity_decodingRules},
for every $\epsilon > 0$
and sufficiently large $n$,
there exists a codebook $\cC = \set{ x^n_1,\dots,x^n_M } \subseteq \cT_0^n(p_X)$
such that $M = \floor{2^{nR}}$, and
for every $x^n \in \cC$, $s^n \in \cS^n$,
and $p_{XX'S} \in \cP_n(\cX \times \cX \times \cS)$,
\begin{multline}
  \card{ \set{ m : (x^n, x_m^n, s^n) \in \cT_0^n(p_{XX'S}) } }
  \\
  \le 3 (n+1)^{\card\cX} 2^{n \max\set{ 0,\, R - I(X S; X')}}
  .
  \label{eq:symTypBound}
\end{multline}

\paragraph{Encoding}
To send the message $m \in [M]$,
the sender feeds the codeword $f(m) \coloneqq x^n_m$ into the channel.

\paragraph{Identification}
Denote the channel output by $Y^n \sim W^n(\cdot | x^n_m, s^n)$,
and suppose the receiver is interested
in the message $m' \in [M]$.
To test whether $m'$ is the sent message or not,
he uses the identification function
\begin{gather}
  \label{eq:identification}
  g(m', Y^n)
  = \ind{
    Y^n \in
    \cD_{m'}
      \setminus
      \cE_{m'}
  }
  \,,
  \\
  \cD_{m'} = \bigcup\nolimits_{p_S \in \cP_S} \cT_\epsilon^n \tup{ p_X \times p_S W | x^n_{m'} },
  \\
  \cE_{m'} = 
  \bigcup_{\ttm \ne m',\ \tts^n,\mathrlap{\ p_{XX'SY} \in \cQ_\delta(p_X, \cP_S, W)}}
  \cT_0^n \tup{ p_{XX'SY} | x^n_{m'}, x^n_{\ttm}, \tts^n }.
\end{gather}
The DI code $(f, g)$ is revealed to all senders and receivers.

\subsection{Error analysis}

\paragraph{Missed identification}
Suppose a message $m$ was sent and $m' = m$.
By~\cite[Theorem~1.2]{kramer2008multi_user_book},
for $s^n \in \cT_0^n(p_S)$,
\begin{align}
  W^n\tup{ \cD_m \,\middle|\, x^n_m, s^n }
  &\ge W^n\tup{ \cT_\epsilon^n(p_X \times p_S W | x^n_m) \,\middle|\, x^n_m, s^n }
  \nonumber
  \\
  &\ge W^n\tup{ \cT_\epsilon^n(p_{XSY} | x^n_m, s^n) \,\middle|\, x^n_m, s^n }
  \\
  &\ge 1 - 2 \card\cX \card\cS \card\cY e^{-n \epsilon^2 \mu}
  \,
\end{align}
where $\mu = \min_{x, s, y \::\: p_{XSY}(x,s,y) > 0} \, p_{XSY}(x,s,y)$,
and
\begin{align}
  &W^n (\cE_m | x^n_m, s^n)
  \le
  \max_{\tts^n}
    \card{ \set{ \ttm : (x^n, x^n_{\ttm}, \tts^n) \in \cT_0^n(p_{XX'S}) } }
  \nonumber\\& \quad
    \max_{\ttm} 
    W^n \tup{ \cT_0^n(p_{XX'SY} | x^n_m, x^n_{\ttm}, \tts^n) \,\middle|\, x^n_m, s^n }
  \\ &
    \le 3 (n+1)^{\card\cX} 2^{n \brack{ R - I(XS; X') + H(Y|XX'S) - H(Y|XS) } }
    \label{eq:typeI.bounds}
  \\ &
    = 3 (n+1)^{\card\cX} 2^{n \brack{ R - I(XS; X') - I(Y; X'|XS) } }
  \\ &
    = 3 (n+1)^{\card\cX} 2^{n \brack{ R - I(XSY; X') }}
    < 3(n+1)^{\card\cX} 2^{-n\alpha}
  \label{eq:typBound}
  \,,
\end{align}
where \cref{eq:typeI.bounds}
follows from \cref{eq:symTypBound}
and~\cite[Theorem~1.2]{kramer2008multi_user_book},
and \cref{eq:typBound} holds by \cref{eq:rate}.
Hence, for some $\beta > 0$, the error is bounded by
\begin{align}
  e(m|m, s^n)
  &= W^n \tup{ G(m, Y^n) = 0 \,\middle|\, x^n_m, s^n }
  \\& \le W^n \tup{ \cD_m^c \,\middle|\, x^n_m, s^n } + W^n \tup{ \cE_m \,\middle|\, x^n_m, s^n }
  \\&
  \le e^{-n \beta}
  \,.
  \label{eq:typeI}
\end{align}

\paragraph{False identification}
Suppose now that $m' \ne m$, and note that
we can partition $\cA = \cD_m \cap \cD_{m'} \cap \cE_m^c \cap \cE_{m'}^c$ into
subsets $\cT_0^n(p_{XX'SY} | x^n_m, x^n_{\ttm}, \tts^n)$, $\ttm \in [M]$, $\tts^n \in \cS^n$.
However, by the definition of $\cE_{m'}$,
either $\cA \cap \cT_0^n(p_{XX'SY}| x^n_{m'}, x^n_{\ttm}, \tts^n) = \emptyset$,
or $p_{XX'SY} \notin \cQ_\delta(p_S, \cP_S, W)$, i.e. $I(X; Y| X'S) > \delta$,
if $p_{Y|X} = p_S W$, $p_{Y|X'} = p_{S'} W$, and $p_S, p_{S'} \in \cP_S$.
Hence, it suffices to bound the probabilities that $Y^n \notin \cA$,
and that the joint type of $x^n_m, x^n_{m'}, s^n$, and $Y^n$
is not in $\cQ_\delta(p_S, \cP_S, W)$. The latter event is implied by
$Y^n \in \cA \cap \cT_0^n(p_{XX'SY}| x^n_{m'}, x^n_m, s^n)$. Thus,
\begin{align}
  &e(m'|m,s^n)
  \nonumber\\
  &= W^n\tup{ G(m', Y^n) = 1 \,\middle|\, x^n_m, s^n }
  \\
  &= W^n\tup{ \cD_{m'} \cap \cE_{m'}^c \,\middle|\, x^n_m, s^n }
  \\
  &\le W^n\tup{ \cD_m^c \cup \cE_m \,\middle|\, x^n_m, s^n }
  \nonumber\\ &\quad
  + W^n\tup{ \cD_m \cap \cD_{m'} \cap \cE_m^c \cap \cE_{m'}^c \,\middle|\, x^n_m, s^n }
  \\
  &\le e(m|m,s^n)
    + \sum_{p_{XX'SY} \mathrlap{\in \cP_n(\cX \times \cX' \times \cS \times \cY) \cap \cQ_\delta^c}}
    W^n\tup{ \cT^n_0(p_{XX'SY}|x^n_{m'},x^n_m \,\middle|\, x^n_m, s^n }
    \nonumber
  \\ &
  \le e^{-n\alpha}
    + (n+1)^{\card\cX^2 \card\cS \card\cY}
    2^{ n \brack{ H(Y|XX'S) - H(Y|X'S) } }
  \\ &
  = e^{-n\alpha}
    + (n+1)^{\card\cX^2 \card\cS \card\cY}
    2^{ - n I(X; Y|X'S) }
  \\ &
  \le e^{-n\alpha}
    + (n+1)^{\card\cX^2 \card\cS \card\cY}
    2^{- n\delta}
  \,.
  \label{eq:typeII}
\end{align}

By \cref{eq:rate}, \cref{eq:typeI}, and \cref{eq:typeII}, the proof is complete.
\qed

\section{Conclusions}

We determined the DI-capacity of a general $K$-MAC
for an average-error criterion, where arbitrarily large codes
are possible.
Furthermore,
we derived upper and lower bounds on the DI-capacity of a K-user multiple-access
channel, with an average input cost constraint,
for a maximal-error criterion.
The bounds coincide if all partial channels are injective, averaged over all
other partial channels.
We illustrate the results at the example of binary adder and multiplier MACs,
where the DI capacity regions strictly include the Shannon capacity regions for
message transmission.
The achievability proof interprets the MAC as a collection of arbitrarily
varying point-to-point channels (AVCs)~\cite{lapidothNarayan1998uncertainty},
where a potentially malicious jammer interferes with
the signal. Therefore, \cref{prop:achiev} yields a new achievability bound for
DI over AVCs under average input constraints.
Further investigation is needed to find exact bounds for non-injective channels,
which could also lead to general tight bounds for DI over AVCs.

\ifblind\else
\section*{Acknowledgement}

The authors
acknowledge the financial support by the Federal Ministry of
Education and Research of Germany (BMBF) in the program of “Souverän. Digital.
Vernetzt.” Joint project 6G-life, project identification number: 16KISK002.
C. Deppe and R. Ferrara were further supported in part by the
BMBF within the national initiative on Post Shannon Communication (NewCom) under the grant 16KIS1005. C. Deppe was also supported by the DFG
within the project DE1915/2-1.

\fi

\clearpage
\printbibliography

\newpage
\listoffixmes

\end{document}